\documentclass[11pt,documentclass,onecolumn]{article}
\usepackage{amsmath}
\usepackage{amsthm}
\usepackage{amssymb}
\usepackage{algorithm}
\usepackage{subcaption}
\usepackage{color}
\usepackage[english]{babel}
\usepackage{graphicx}
\usepackage{wrapfig,epsfig}
\usepackage{psfrag}
\usepackage{epstopdf}
\usepackage{url}
\usepackage{graphicx}
\usepackage{color}
\usepackage{epstopdf}
\usepackage{algpseudocode}
\usepackage{tikz}
\usepackage{hyperref}
\hypersetup{colorlinks=true}
\usetikzlibrary{arrows}

\usepackage[margin=1in]{geometry}
\author{Daniel Kane \and Sushrut Karmalkar \and Eric Price}

\title{Robust polynomial regression up to the information theoretic limit}

\newtheorem{theorem}{Theorem}[section]
\newtheorem{lemma}[theorem]{Lemma}
\newtheorem{definition}[theorem]{Definition}

\newtheorem{corollary}[theorem]{Corollary}

\newcommand{\abs}[1]{|#1|}
\newcommand{\tabs}[1]{\left|#1\right|}
\newcommand{\norm}[1]{\|#1\|}

\newcommand{\wh}{\widehat}
\newcommand{\wt}{\widetilde}
\newcommand{\eps}{\epsilon}

\newcommand{\R}{\mathbb{R}}

\newcommand{\RN}[1]{%
  \textup{\uppercase\expandafter{\romannumeral#1}}%
}

\DeclareMathOperator*{\E}{\mathbb{E}}
\DeclareMathOperator*{\median}{median}
\DeclareMathOperator*{\mean}{mean}
\DeclareMathOperator*{\argmin}{arg\,min}

\DeclareMathOperator{\poly}{poly}

\DeclareMathOperator{\obj}{obj}

\usepackage{lineno}

\usepackage{bbm}

\usepackage{etoolbox}

\makeatletter

\newcommand{\define}[4][ignore]{%
  \ifstrequal{#1}{ignore}{}{
  \@namedef{thmtitle@#2}{#1}}%
  \@namedef{thm@#2}{#4}%
  \@namedef{thmtypen@#2}{lemma}%
  \newtheorem{thmtype@#2}[theorem]{#3}%
  \newtheorem*{thmtypealt@#2}{#3~\ref{#2}}%
}

\newcommand{\state}[1]{%
  \@namedef{curthm}{#1}
  \@ifundefined{thmtitle@#1}{
  \begin{thmtype@#1}
    }{
  \begin{thmtype@#1}[\@nameuse{thmtitle@#1}]
  }
    \label{#1}
    \@nameuse{thm@#1}
  \end{thmtype@#1}
  \@ifundefined{thmdone@#1}{
  \@namedef{thmdone@#1}{stated}%
  }{}
}

\newcommand{\restate}[1]{%
  \@namedef{curthm}{#1}
  \@ifundefined{thmtitle@#1}{
    \begin{thmtypealt@#1}
    }{
  \begin{thmtypealt@#1}[\@nameuse{thmtitle@#1}]
  }
    \@nameuse{thm@#1}
  \end{thmtypealt@#1}
  \@ifundefined{thmdone@#1}{
  \@namedef{thmdone@#1}{stated}%
  }{}
}

\newcommand{\thmlabel}[1]{
  \@ifundefined{thmdone@\@nameuse{curthm}}{\label{#1}
    }{\tag*{\eqref{#1}}}
}
\makeatother

\begin{document}
\maketitle
\begin{abstract}
  We consider the problem of \emph{robust polynomial regression},
  where one receives samples $(x_i, y_i)$ that are usually within
  $\sigma$ of a polynomial $y = p(x)$, but have a $\rho$ chance of
  being arbitrary adversarial outliers.  Previously, it was known how
  to efficiently estimate $p$ only when $\rho < \frac{1}{\log d}$.  We
  give an algorithm that works for the entire feasible range of
  $\rho < 1/2$, while simultaneously improving other parameters of the
  problem.  We complement our algorithm, which gives a factor 2
  approximation, with impossibility results that show, for example,
  that a $1.09$ approximation is impossible even with infinitely many
  samples.
\end{abstract}

\section{Introduction}

Polynomial regression is the problem of finding a polynomial that
passes near a collection of input points.  The problem has been
studied for 200 years with diverse
applications~\cite{gergonne1974application}, including computer
graphics~\cite{PrattAlgSurfaces1987}, machine
learning~\cite{kalai2008agnostically}, and
statistics~\cite{macneill1978properties}.  As with linear regression,
the classic solution to polynomial regression is least squares, but
this is not robust to outliers: a single outlier point can perturb the
least squares solution by an arbitrary amount.  Hence finding a
``robust'' method of polynomial regression is an important question.

A version of this problem was formalized by Arora and Khot~\cite{AK03}.  We want to
learn a degree $d$ polynomial $p: [-1, 1] \to \R$, and we observe
$(x_i, y_i)$ where $x_i$ is drawn from some measure $\chi$ (say,
uniform) and $y_i = p(x_i) + w_i$ for some noise $w_i$.  Each sample
is an ``outlier'' independently with probability at most $\rho$; if
the sample is not an outlier, then $\abs{w_i} \leq \sigma$ for some
$\sigma$.  Other than the random choice of outlier locations, the
noise is adversarial.  One would like to recover $\wh{p}$ with
$\norm{p - \wh{p}}_\infty \leq C\sigma$ using as few samples as
possible, with high probability.

In other contexts an $\ell_\infty$ requirement on the input noise
would be a significant restriction, but outlier tolerance makes it
much less so.  For example, independent bounded-variance noise fits in
this framework: by Chebyshev's inequality, the framework applies with outlier chance
$\rho = \frac{1}{100}$ and $\sigma = 10\E[w_i^2]^{1/2}$, which gives
the ideal result up to constant factors.  At the same time, the
$\ell_\infty$ requirement on the input allows for the strong
$\ell_\infty$ bound on the output.

Arora and Khot~\cite{AK03} showed that $\rho < 1/2$ is
information-theoretically necessary for non-trivial recovery, and
showed how to solve the problem using a linear program when
$\rho = 0$.  For $\rho > 0$, the RANSAC~\cite{RANSAC} technique of
fitting to a subsample of the data works in polynomial time when
$\rho \lesssim \frac{\log d}{d}$.  Finding an efficient algorithm for
$\rho \gg \frac{\log d}{d}$ remained open until recently.

Last year, Guruswami and Zuckerman~\cite{GZ16} showed how to solve the
problem in polynomial time for $\rho$ larger than $\frac{\log d}{d}$. Unfortunately, their result falls short of the ideal in several ways: it needs a low
outlier rate $(\rho < \frac{1}{\log d})$, a bounded signal-to-noise ratio
$(\frac{\norm{p}_\infty}{\sigma} < d^{O(1)})$, and has a
super-constant approximation factor
$(\norm{p - \wh{p}}_\infty \lesssim \sigma \left(1 +
  \frac{\norm{p}_\infty}{\sigma}\right)^{0.01})$.  It uses
$O(d^2 \log^c d)$ samples from the uniform measure, or $O(d \log^c d)$
samples from the Chebyshev measure $\frac{1}{\sqrt{1-x^2}}$.

These deficiencies mean that the algorithm doesn't work when all the
noise comes from outliers ($\sigma = 0$); the low outlier rate
required also leads to a superconstant factor loss when reducing from
other noise models.

In this work, we give a simple algorithm that avoids all these
deficiencies: it works for all $\rho < 1/2$; it has no restrictions on
$\sigma$; and it gets a constant approximation factor $C = 2 + \eps$.
Additionally, it only uses $O(d^2)$ samples from the uniform measure,
without any log factors, and $O(d \log d)$ samples from the Chebyshev
measure. We also give lower bounds for the approximation factor $C$, indicating that one cannot achieve a $1 +\eps$ approximation by showing that  $C > 1.09$. Our lower bounds are not the best possible, and it is possible that the true lower bounds are much better.
\subsection{Algorithmic results}

The problem formulation has randomly placed outliers, which is needed
to ensure that we can estimate the polynomial locally around any point
$x$.  We use the following definition to encapsulate this requirement,
after which the noise can be fully adversarial:

\begin{definition}
  The size $m$ \emph{Chebyshev partition} of $[-1, 1]$ is the set of
  intervals $I_j = [\cos \frac{\pi j}{m}, \cos \frac{\pi (j-1)}{m}]$
  for $j \in [m]$.

  We say that a set $S$ of samples $(x_i, y_i)$ is ``$\alpha$-good''
  for the partition if, in every interval $I_j$, strictly less than an
  $\alpha$ fraction of the points $x_i \in I_j$ are outliers. 
\end{definition}

The size $m$ Chebyshev partition is the set of intervals between
consecutive extrema of the Chebyshev polynomial of degree $m$.
Standard approximation theory recommends sampling functions at the
roots of the Chebyshev polynomial, so intuitively a good set of
samples will allow fairly accurate estimation near each of these
``ideal'' sampling points. All our algorithms will work for any set of $\alpha$-
good samples, for $\alpha < \frac{1}{2}$ and sufficiently large $m$. The sample complexities are
then what is required for $S$ to be good with high probability.

We first describe two simple algorithms that do not quite achieve the goal.  We then describe how to black-box combine their results to
get the full result.

\paragraph{L1 regression.} Our first result is that $\ell_1$ regression \emph{almost}
solves the problem: it satisfies all the requirements, except that the
resulting error $\norm{\wh{p} - p}$ is bounded in $\ell_1$ not
$\ell_\infty$:
\define{lem:l1reg}{Lemma}{%
  Suppose the set $S$ of samples is $\alpha$-good for the size $m = O(d)$
  Chebyshev partition, for constant $\alpha < \frac{1}{2}$.  Then the result $\wh{p}$ of $\ell_1$
  regression
  \[
  \argmin_{\substack{\text{degree-}d\\\text{polynomial }\wh{p}}} \sum_{i=1}^n \abs{I_j} \mean_{x_i \in I_j} \abs{y_i - \wh{p}(x_i)}
  \]
  satisfies $\norm{\wh{p} - p}_1 \leq O_{\alpha}(1) \cdot \sigma$.
}
\state{lem:l1reg}
This has a nice parallel to the $d=1$ case, where $\ell_1$ regression
is the canonical robust estimator.

As shown by the Chebyshev polynomial in Figure~\ref{fig:chebyshev}, $\ell_1$ regression might not give a good solution to the $\ell_{\infty}$ problem. However, converting to an $\ell_\infty$ bound loses at most an $O(d^2)$
factor. This means this lemma is already useful for an $\ell_\infty$ bound: one
of the requirements of~\cite{GZ16} was that $\norm{p}_\infty \leq
\sigma d^{O(1)}$.  We can avoid this by first computing the $\ell_1$
regression estimate $\wh{p}^{(\ell_1)}$ and then applying~\cite{GZ16}
to the residual $p - \wh{p}^{(\ell_1)}$, which always satisfies the
condition.

\begin{figure}
  \centering
  \includegraphics[width=0.48\textwidth]{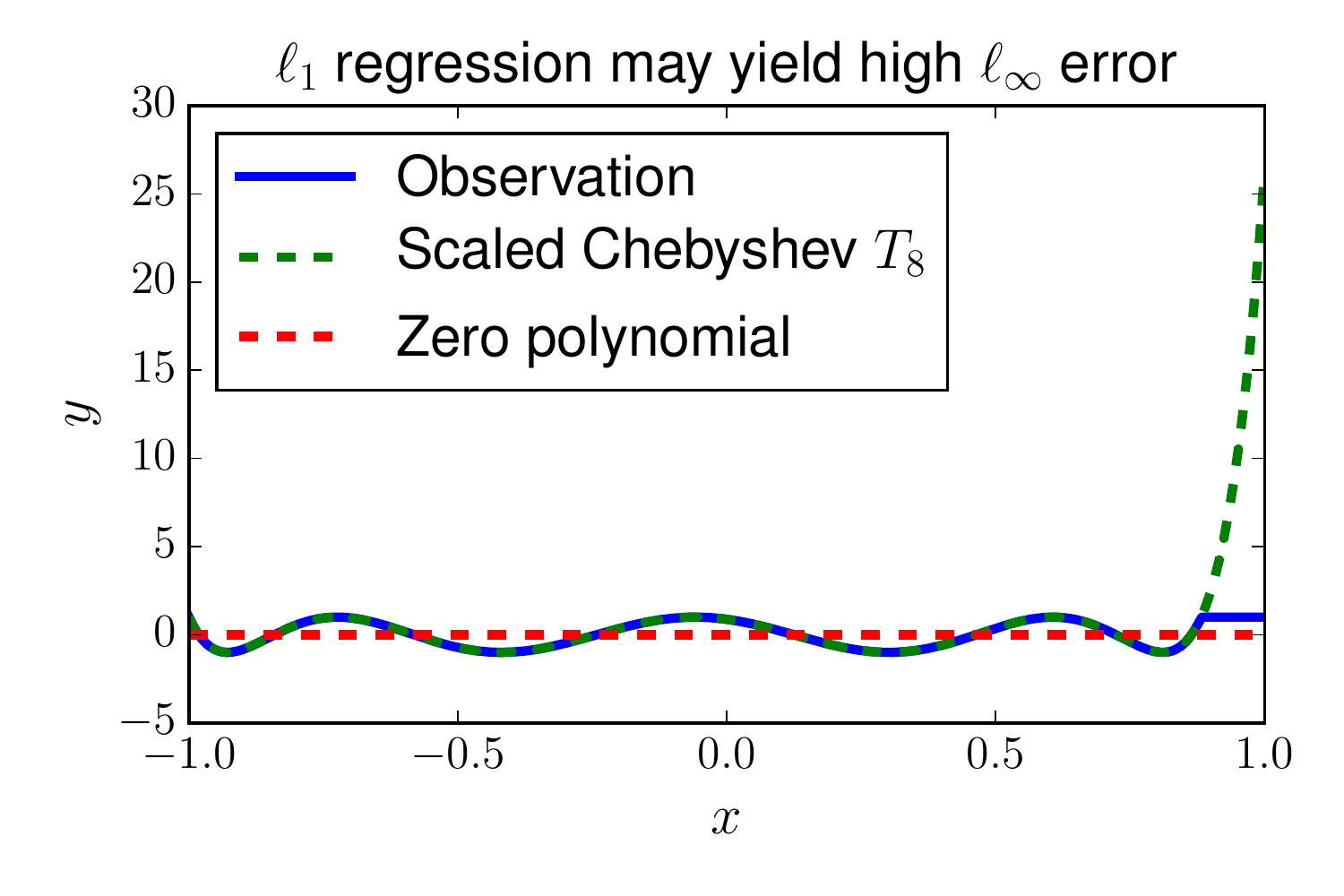}
  \captionsetup{aboveskip=-5pt}
  \caption{The blue observations are within $\sigma = 1$ of the red
    zero polynomial at each point, but are closer in $\ell_1$ to the
    green Chebyshev polynomial.  This demonstrates that $\ell_1$
    regression does not solve the $\ell_\infty$ problem even for $\rho
    = 0$.}
  \label{fig:chebyshev}
\end{figure}

\paragraph{Median-based recovery.} Our second result is for a
median-based approach to the problem: take the median $\wt{y}_j$ of
the $y$ values within each interval $I_j$.  Because the median is
robust, this will lie within the range of inlier $y$ values for that
interval.  We don't know which $x \in I_j$ it corresponds to, but this
turns out not to matter too much: assigning each $\wt{y}_j$ to an
arbitrary $\wt{x}_j \in I_j$ and applying non-robust $\ell_\infty$
regression gives a useful result.

\define{lem:medianrecov}{Lemma}{%
  Let $\eps, \alpha < \frac{1}{2}$, and suppose the set $S$ of samples
  is $\alpha$-good for the size $m = O(d/\eps)$ Chebyshev partition.
  Let $\wt{x}_j \in I_j$ be chosen arbitrarily, and let
  $\wt{y}_j = \median_{x_i \in I_j} y_i$.  Then the degree $d$
  polynomial $\wh{p}$ minimizing
  \[
    \max_{j \in [m]} \abs{\wh{p}(\wt{x}_j) - \wt{y}_j}
  \]
  satisfies $\norm{\wh{p} - p}_\infty \leq (2+\eps)\sigma + \eps \norm{p}_\infty$.
}
\state{lem:medianrecov}

Without the additive $\eps \norm{p}_\infty$ term, which comes from not
knowing the $x \in I_j$ for $\wt{y}_j$, this would be exactly what we want.
If $\sigma \ll \norm{p}_\infty$, however, this is not so good---but it
still makes progress.

\paragraph{Iterating the algorithms.} While neither of the above
results is sufficient on its own, simply applying them iteratively on
the residual left by previous rounds gives our full theorem.  The
result of $\ell_1$ regression is a $\wh{p}^{(0)}$ with
$\norm{p - \wh{p}^{(0)}}_\infty = O(d^2 \sigma)$.  Applying the
median recovery algorithm to the residual points
$(x_i, y_i - \wh{p}^{(0)}(x_i))$, we will get $\wh{p}'$ so that
$\wh{p}^{(1)} = \wh{p}^{(0)} + \wh{p}'$ has
\[
  \norm{\wh{p}^{(1)} - p}_\infty \leq (2 + \eps) \sigma + \eps \cdot O(d^2 \sigma)
\]
If we continue applying the median recovery algorithm to the remaining
residual, the $\ell_\infty$ norm of the residual will continue to
decrease exponentially.  After $r=O(\log_{1/\eps} d)$ rounds we will
reach
\[
  \norm{\wh{p}^{(r)} - p}_\infty \leq (2+4\eps) \sigma
\]
as desired\footnote{We remark that one could skip the initial $\ell_1$
  minimization step, but the number of rounds would then depend on
  $\log (\frac{\norm{p}_\infty}{\sigma})$, which could be
  unbounded. Note that this only affects running time and not the
  sample complexity.}.  Since each method can be computed efficiently
with a linear program, this gives our main theorem:

\define{thm:main}{Theorem}{%
  Let $\eps, \alpha < \frac{1}{2}$, and suppose the set $S$ of samples
  is $\alpha$-good for the size $m = O(d/\eps)$ Chebyshev partition.
  The result $\wh{p}$ of Algorithm~\ref{alg:main} is a degree $d$
  polynomial satisfying
  $\norm{\wh{p} - p}_\infty \leq (2+\eps)\sigma$.  Its running time is
  that of solving $O(\log_{1/\eps} d)$ linear programs, each with
  $O(d)$ variables and $O(\abs{S})$ constraints.  } \state{thm:main}

We now apply this to the robust polynomial regression problem, where we
receive points $(x_i, y_i)$ such that each $x_i$ is drawn from some
distribution, and with probability $\rho$ it is an outlier.  An
adversary then picks the $y_i$ such that, for each non-outlier $i$,
$\abs{y_i - p(x_i)} \leq \sigma$.  We observe the following:
\begin{corollary}\label{cor:randomerrors}
  Consider the robust polynomial regression problem for constant
  outlier chance $\rho < 1/2$, with points $x_i$ drawn from the
  Chebyshev distribution $D_c(x) \sim \frac{1}{\sqrt{1-x^2}}$.  Then
  $O(\frac{d}{\eps} \log \frac{d}{\delta\eps})$ samples suffice to
  recover with probability $1-\delta$ a degree $d$ polynomial $\wh{p}$
  with
  \[
    \max_{x \in [-1, 1]} \abs{p(x) - \wh{p}(x)} \leq (2 + \eps)\sigma.
  \]
  If $x_i$ is drawn from the uniform distribution instead, then
 $O(\frac{d^2}{\eps^2}\log \frac{1}{\delta})$ samples suffice for the
  same result.  In both cases, the recovery time is polynomial in the
  sample size.
\end{corollary}

\subsection{Impossibility results}

We show limitations on improving any of the three parameters used in
Corollary~\ref{cor:randomerrors}: the sample complexity, the
approximation factor, and the requirement that $\rho < 1/2$.

\paragraph{Sample Complexity.} Our result uses $O(d^2)$ samples from
the uniform distribution and $O(d \log d)$ samples from the Chebyshev
distribution.  We show in Section~\ref{sec:sampcomp} that both results
are tight. In particular, we show that it is impossible to get any
constant approximation with $o(d^2)$ samples from the uniform
distribution or $o(d \log d)$ samples from the Chebyshev distribution.

\paragraph{Approximation factor.} Our approximation factor is
$2 +\eps$.  Even with infinitely many samples and no outliers, can one
do better than a $2$-approximation for $\ell_\infty$ regression?  For
comparison, in $\ell_2$ regression a $1 + \eps$ approximation is
possible.  For these lower bounds, it is convenient to consider the
special case where the values $y_i$ are $y(x_i)$ for a function $y$
with $\norm{p - y}_\infty \leq \sigma$.  We show two lower bounds
related to this question.

First, we show that $\ell_\infty$ projection---that is, the algorithm
that minimizes $\norm{y - \wh{p}}_\infty$ over degree-$d$ polynomials
$\wh{p}$---can have error arbitrarily close to $2\sigma$.  Since our
algorithm is an outlier-tolerant version of $\ell_\infty$ projection,
we should not expect to perform better.

Second, we show that no proper learning algorithm can achieve a
$1+\eps$ guarantee.  In particular, we show for $d=2$ that any
algorithm with more than $2/3$ success probability must have
$C > 1.09$; this is illustrated in Figure~\ref{fig:lowerbound}.  For
general $d$, we show that $C > 1 + \Omega(1/d^3)$.

\begin{figure}
  \centering
  \includegraphics[width=0.48\textwidth]{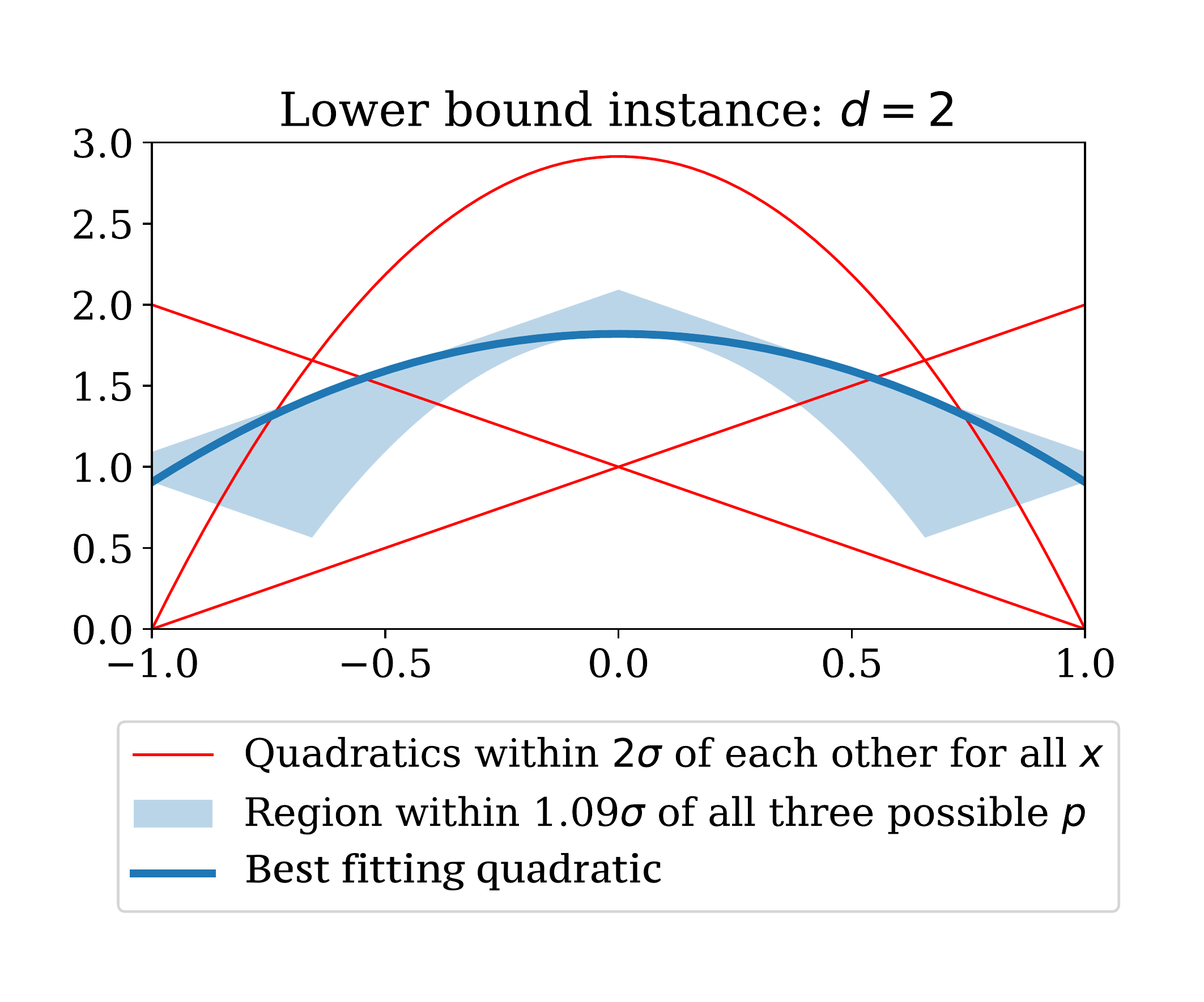}
  \caption{The lower bound for $1.09$-approximations via any recovery
    algorithm.  We give three quadratics that all lie within $2\sigma$
    of each other for all $x$, and hence the observed $y(x)$ can be
    identical regardless of which quadratic is $p$.  The region
    containing points within $1.09\sigma$ of all three options just
    barely doesn't contain a quadratic itself.}
  \label{fig:lowerbound}
\end{figure}

\paragraph{List decoding.}  The $\rho < 1/2$ requirement is obviously
required for uniquely decoding $p$, since for $\rho \geq 1/2$ the
observations could come from a mixture of two completely different
polynomials.  But one could hope for a list decoding version of the
result, outputting a small set of polynomials such that one is close
to the true output.  Unfortunately, \cite{AK03} showed that for a
$C$-approximate algorithm, such a set would require size at least
$e^{\Omega(\sqrt{d} / C)}$ event for $\rho \geq 1/2$.  We improve this
lower bound to $e^{\Omega(d / C)}$.  At least for constant $C$ and
$\rho$, our result is tight: Theorem~\ref{thm:main} implies that if no
samples were outliers, then $m = O(d)$ samples would suffice.  Hence
picking $m$ samples will result in Algorithm 2 outputting a polynomial consistent with the data with
$(1-\rho)^m = e^{-O(d)}$ probability. Repeating this would give a
set of size $e^{O(d)}$ that works with high probability.

\subsection{Related Work}

In addition to the work of~\cite{AK03,GZ16} discussed above, several
papers have looked at similar problems.  When $\sigma = 0$, the
problem becomes the standard one of Reed-Solomon decoding with
relative Hamming distance $\rho$.  Efficient algorithms such as
Berlekamp-Massey~\cite{WB86} give unique decoding for all
$\rho < 1/2$.  For $\rho > 1/2$, while unique decoding is impossible,
a polynomial size list decoding is possible---with sufficiently many
samples---for all $\rho < 1$~\cite{GS98}.

Other work has looked at robust estimation for distributions.  In the
field of robust statistics (see, e.g.,~\cite{huber}), as well as some
recent papers in theoretical computer science
(e.g. \cite{lai,dklms,csv}), one would like to estimate properties of
a distribution from samples that have a $\rho$ chance of being
adversarial.  In some such cases, list decoding for $\rho > 1/2$ is
possible~\cite{csv}.


\section{Preliminaries}

 For a function $f: [-1, 1] \to \R$ and interval $I \subseteq [-1, 1]$, we define the $\ell_{q}$ norm on the interval to be
\[
\norm{f}_{I,q} := \left(\int_{I_k} \abs{f(x)}^qdx\right)^{1/q},
\]
where $\norm{f}_{I,\infty} := \sup_{x \in I} \abs{f(x)}$.  We also
define the overall $\ell_q$ norm $\norm{f}_q := \norm{f}_{[-1,1],
  q}$. We will need the following consequence of a generalization due
to Nevai~\cite{N79} of Bernstein's inequality from $\ell_{\infty}$
to $\ell_{q}$. The proof of this lemma is in Appendix A.

\define{lem:polyconstant}{Lemma}{%
  Let $p$ be a degree $d$ polynomial.  Let $I_1, \dotsc, I_m$
  partition $[-1, 1]$ between the Chebyshev extrema $\cos \frac{\pi
    j}{m}$, for some $m \geq d$.  Let $r: [-1,1]\to \R$ be piecewise
  constant, so that for each $I_k$ there exists an $x^*_k \in I_k$
  with $r(x) = p(x^*_k)$ for all $x \in I_k$.  Then there exists a
  universal constant $C$ such that, for any $q \geq 1$,
  \[
  \norm{p - r}_q \leq C\frac{d}{m} \norm{p}_q.
  \]
}
\state{lem:polyconstant}

We recall the definition of the Chebyshev polynomials $T_d(x)$, which
we will use extensively.
\begin{definition}
  The Chebyshev polynomials of the first kind $T_d(x)$ are defined by
  the following recurrence: $T_0(x) = 1, T_1(x) = x$ and
  $T_{n+1}(x) = 2xT_n(x)-T_{n-1}(x)$.  They have the property that
  $T_d(\cos \theta) = \cos(d \theta)$ for all $\theta$.
      \end{definition}

\section{$\ell_1$ regression}

\begin{lemma}\label{lem:l1avg}
	Suppose $m \geq C\frac{d}{\eps}$ for a large enough constant $C$.
	Then, for any set of samples $x_1, \dotsc, x_n$ with all $S_i = \{j
	\mid x_j \in I_i\}$ nonempty, and any polynomial $p$ of degree $d$,
	\[
	\sum_i \frac{\abs{I_i}}{\abs{S_i}} \sum_{j \in S_i} \abs{p(x_j)} =
	(1 \pm \eps)\norm{p}_1.
	\]
\end{lemma}
\begin{proof}
  When all $\abs{S_i} = 1$, this is a restatement of
  Lemma~\ref{lem:polyconstant} for $q=1$.  Otherwise, the LHS is the
  expectation of the $\abs{S_i} = 1$ case, when randomly drawing a
  single $j$ in each $S_i$.
\end{proof}

We will show a more precise statement than Lemma~\ref{lem:l1reg},
which allows for a weaker $\ell_1$ version of the $\alpha$-good
requirement on the samples:
\begin{definition}\label{def:l1close}
  For a set of samples $(x_1, y_1), \dotsc, (x_n, y_n)$, and the
  Chebyshev partition $I_1, \dotsc, I_m$, define
  $S_j = \{i \mid x_i \in I_j\}$.  We say that the samples are
  $(\alpha, \sigma)$ close to a polynomial $p$ in $\ell_1$ if, for
  \[
    e_j := \min_{S' \subset S_j, \abs{S'} \leq (1-\alpha) \abs{S_j}} \max_{j \in S'} \abs{p(x_j) - y_j},
  \]
  we have
  \[
    \sum_{j \in [m]} \abs{I_j} e_j \leq \sigma.
  \]
\end{definition}
If the samples are $\alpha$-good, then each $e_j \leq \sigma$, so
$\sum \abs{I_j} e_j \leq (\sum \abs{I_j}) (\max e_j) \leq 2 \sigma$
and the samples are $(\alpha, 2\sigma)$ close to $p$ in $\ell_1$.

We now state Lemma~\ref{lem:l1regprecise} which is a more precise
statement of Lemma~\ref{lem:l1reg}.
\begin{lemma}\label{lem:l1regprecise}
  Let $\alpha < 1/2$, and $m \geq C\frac{d}{\eps}$ for a large enough
  constant $C$ and some $\eps \leq (1 - 2\alpha)/4$.  Then, given any
  samples $(\alpha, \sigma)$ close to $p$ in $\ell_1$, the degree $d$
  polynomial solution to the $\ell_1$ regression problem
	\[
	\widehat{p} = \argmin_{q} \sum_{j\in [m]} \frac{\abs{I_j}}{\abs{S_j}}\sum_{i \in S_j} \abs{y_i - q(x_i)}
	\]
	has
	\[
	\norm{\wh{p} - p}_1 \leq \frac{2\sigma}{1-2\alpha}.
	\]
\end{lemma}

\begin{proof}
	In each interval, we consider the $\alpha \abs{S_j}$ coordinates
	maximizing $\abs{y_i - p(x_i)}$ to be `bad', and the rest to be
	`good'. Let the set of `bad' and `good' coordinates in $S_j$ be denoted by $B_j$ and $G_j$ respectively, and define  $\obj(f) = \sum_{j\in [m]} \frac{\abs{I_j}}{\abs{S_j}}\sum_{i \in S_j} \abs{y_i - f(x_i)}$. By definition $\obj(\wh{p}) \leq \obj(p)$. This gives us 
	\begin{align*}
	0 &\geq \obj(\widehat{p}) - \obj(p)\\
	& = \sum_{j \in [m]}\frac{\abs{I_j}}{\abs{S_j}}\left(\sum_{i \in S_j} \abs{y_i - \widehat{p}(x_i)} - \sum_{i \in S_j} \abs{y_i - p(x_i)}\right)\\
	&= \sum_{j \in [m]}\frac{\abs{I_j}}{\abs{S_j}}\left(\left(\sum_{i \in G_j} \abs{y_i - \widehat{p}(x_i)}  - \abs{y_i - p(x_i)}\right) + \left( \sum_{i \in B_j} \abs{y_i - \widehat{p}(x_i)} - \abs{y_i - p(x_i)}\right)\right).
	\end{align*}
	From the triangle inequality, we have $ |y_i - \widehat{p}(x_i)| \geq |p(x_i) - \widehat{p}(x_i)| -   |y_i - p(x_i)|$ and $|y_i - \widehat{p}(x_i)|  - |y_i - p(x_i)| \geq -|p(x_i) - \widehat{p}(x_i)| $. This gives
	\begin{align*}
	0 & \geq \sum_{j \in [m]}\frac{\abs{I_j}}{\abs{S_j}}\left(\sum_{i \in G_j}(\abs{p(x_i) - \widehat{p}(x_i)} - 2\abs{y_i - p(x_i)}) - \sum_{i \in B_j} \abs{p(x_i) - \widehat{p}(x_i)}\right)
	\end{align*}
	Now, recall that our samples are $(\alpha, \delta)$ close to
        $p$ in $\ell_1$. This means (by Definition~\ref{def:l1close}) that for any
        `good' $i$, $\abs{y_i - p(x_i)} \leq e_j$ for a set of $e_j$'s
        that satisfy $\sum |I_j|e_j \leq \sigma$.  Therefore, for these $e_j$,
	\begin{align*}
	0 &\geq \sum_{j \in [m]}\frac{\abs{I_j}}{\abs{S_j}}\left(\sum_{i \in G_j}\abs{p(x_i) - \widehat{p}(x_i)} - \sum_{i \in B_j} \abs{p(x_i) - \widehat{p}(x_i)}\right) - \sum_{j \in [m]}\frac{\abs{I_j}}{\abs{S_j}}\left( \sum_{i \in G_j} \abs{y_i - p(x_i)}\right) \\
	& \geq \sum_{j \in [m]}\frac{\abs{I_j}}{\abs{S_j}}\left(\sum_{i \in G_j}\abs{p(x_i) - \widehat{p}(x_i)} - \sum_{i \in B_j} \abs{p(x_i) - \widehat{p}(x_i)}\right) - \sum_{j \in [m]}\frac{\abs{I_j}}{\abs{S_j}}\left((1-\alpha)|S_j| e_j\right)
	\end{align*}
	Using Lemma~\ref{lem:l1avg} on $p-\widehat{p}$ we now get 
	\begin{align*}
	0 \geq (1-\alpha)(1-\eps) \|p-\widehat{p}\|_1 - \alpha (1+\eps) \|p - \widehat{p}\|_1 -(1-\alpha)\sigma
	\end{align*}
	or
	\[
	(1-2\alpha-\eps)\norm{p-\widehat{p}}_1 \leq (1-\alpha)\sigma,
	\]
	Hence
	\[
	\norm{p-\widehat{p}}_1 \leq \frac{(1-\alpha)\sigma}{1-2\alpha-\eps}.
	\]
	For $\eps \leq \frac{1 - 2\alpha}{2}$, this gives the desired
	\[
	\norm{p-\widehat{p}}_1 \leq \frac{2\sigma}{1 - 2\alpha}.
	\]
	
\end{proof}

\section{$\ell_\infty$ regression}

Given a set of $\alpha < \frac{1}{2}$-good samples $S$, our goal is to find a degree $\wh{p}$
polynomial $q$ with $\norm{\wh{p}-p}_\infty \leq (2+\epsilon)\sigma$. We
start by proving Lemma~\ref{lem:medianrecov}, which we restate below for
clarity:
\begin{algorithm}
	\caption{Refinement method, analyzed in Lemma~\ref{lem:medianrecov} for $r=0$}\label{alg:refine}
	\begin{algorithmic}[1]
		\Procedure{$\text{Refine}(S = \{(x_i,y_i)\},\wh{p})$}{}
		
		\State $\widetilde{y}_j \gets \text{median}_{x_i \in I_j} y_i- \wh{p}(x_i)$
		\State Choose arbitrary $\wt{x}_j \in I_j$
		\State Fit degree $d$ polynomial $r$ minimizing $\|r(\wt{x}_j) - \widetilde{y}_j\|_{\infty}$
		\State $\widehat{p}' \gets \widehat{p} + r$
		\State \textbf{return} $\widehat{p}'$
		
		\EndProcedure
	\end{algorithmic}
\end{algorithm}
\begin{algorithm}
	\caption{Complete recovery procedure}\label{alg:main}
	\begin{algorithmic}[1]
		\Procedure{$\text{Approx}(S)$}{}
		\State $\widehat{p}^{(0)}\gets$ result of $\ell_1$ regression
		\For {$i \in [0, O(\log_{1/\eps} d)]$}
			\State $\widehat{p}^{(i+1)} \gets \text{REFINE}(S, \widehat{p}^{(i)})$
		\EndFor
		\State \textbf{return} $\widehat{p}^{( O(\log_{1/\eps} d))}$
	\EndProcedure
	\end{algorithmic}
      \end{algorithm}

\restate{lem:medianrecov}
\begin{proof} 
  Since more than half the points in any interval $I_j$ are such that
  $\abs{y_j -p(x_j)} \leq \sigma$, and $p$ is continuous, there must
  exist an $x'_j \in I_j$ satisfying
  \begin{align}
    |\wt{y}_j - p(x'_j)| \leq \sigma.\label{eq:maxgap}
  \end{align}

  We now define three piecewise-constant functions, $r(x), \wh{r}(x)$,
  and $\wt{r}(x)$, to be such that within each interval $I_j$ we have
  $r(x) = p(x'_j)$, $\wh{r}(x) = \wh{p}(\wt{x}_j)$, and
  $\wt{r}(x) = \wt{y}_j$.  By Lemma~\ref{lem:polyconstant}, for our
  choice of $m$ we have
  \begin{align}
    \norm{p - r}_\infty \leq \eps \norm{p}_\infty\text{~~~~~~~~and~~~~~~~~} \norm{\wh{p} - \wh{r}}_\infty \leq \eps \norm{\wh{p}}_\infty \label{eq:2}.
  \end{align}
  We also have by~\eqref{eq:maxgap} that $\norm{r - \wt{r}}_\infty \leq \sigma$, so by the triangle inequality
  \begin{align}
    \norm{p - \wt{r}}_{\infty} \leq \sigma + \eps\norm{p}_\infty.\label{eq:3}
  \end{align}
  Now, our choice of $\wh{p}$ ensures that
  \[
    \norm{\wh{r} - \wt{r}}_\infty = \max_{j \in [m]}\abs{\wh{p}(\wt{x}_j) - \wt{y}_j} \leq \max_{j \in [m]} \abs{p(\wt{x}_j) - \wt{y}_j} \leq \norm{p - \wt{r}}_\infty \leq \sigma + \eps\norm{p}_\infty.
  \]
  Combining with~\eqref{eq:2} and~\eqref{eq:3} gives by the triangle inequality that
  \begin{align}
    \norm{\wh{p} - p}_\infty \leq 2\sigma + 2\eps\norm{p}_\infty + \eps \norm{\wh{p}}_\infty.\label{eq:4}
  \end{align}
  To finish the proof, we just need a bound on $\norm{\wh{p}}_\infty$.
  Note that~\eqref{eq:4} implies
  \[
    \norm{\wh{p}}_\infty \leq 2\sigma + (1+2\eps)\norm{p}_\infty + \eps \norm{\wh{p}}_\infty,
  \]
  and since $\eps \leq 1/2$, this implies
  \[
    \norm{\wh{p}}_\infty \leq 4\sigma + (2+4\eps)\norm{p}_\infty.
  \]
  Plugging into~\eqref{eq:4} and rescaling $\eps$ down by a constant
  factor gives the result.
\end{proof} 

Before we prove Theorem~\ref{thm:main}, we briefly describe the
algorithm. Algorithm~\ref{alg:main} first sets its initial estimate to
the polynomial produced by $\ell_1$ regression. It then continues to
refine the estimate it has by using Algorithm~\ref{alg:refine} for
$O(\log_{1/\eps} d)$ iterations.


\restate{thm:main}
\begin{proof}
	
  Our algorithm proceeds by iteratively improving a polynomial
  approximation $\widehat{p}$ to $p$, so that
  $\norm{p-\widehat{p}}_\infty$ improves at each stage.  Our notation
  will be borrowed from the notation in the statements of
  Algorithms~\ref{alg:main} and~\ref{alg:refine}. Let
  $\widehat{p}^{(t)}(x)$ be the $t^{\text{th}}$ estimate of $p$ found by
  Algorithm~\ref{alg:main} and let $e_t(x) = (p-\widehat{p}^{(t)})(x)$
  be the error of the $t^{\text{th}}$ estimate $\widehat{p}^{(t)}$. $r_t$ is
  the polynomial $r$ found by the $t^{\text{th}}$ iteration of
  Algorithm~\ref{alg:refine}, i.e.
  $r_t = \argmin_{r} \|r(\tilde{x}_j) - \tilde{y}_j\|_{\infty}$ where
  $\widetilde{y}_j = \text{median}_{x_i \in I_j} y_i-
  \widehat{p}^{(t)}(x_i)$, and the minimum is taken over all degree
  $d$ polynomials. Lemma~\ref{lem:medianrecov} now
  implies
  \[
    \|r_t(x)-e_t(x)\|_{\infty} \leq (2+\eps)\sigma + \eps \|e_t(x)\|_{\infty}.
  \]
  Observe that
  $r_t(x) - e_t(x) = (\widehat{p}^{(t+1)}(x) - \widehat{p}^{(t)}(x)) -
  (p(x) - \widehat{p}^{(t)}(x)) = -e_{t+1}(x)$, which gives
  us
  \[\|e_{t+1}(x)\|_{\infty} \leq (2+\eps)\sigma +
    \epsilon\|e_t(x)\|_{\infty}.\] Proceeding by induction and using
  the geometric series formula we see
  \[\|p-\widehat{p}^{(t)}\|_{\infty} \leq
    \frac{2+\eps}{1-\epsilon}\sigma +
    \epsilon^t\|e_0\|_{\infty} \leq (2+4\eps)\sigma +
    \epsilon^t\|e_0\|_{\infty}. \] Rescaling
  $\epsilon$, we see that in a number of iterations logarithmic in the
  quality of our initial solution, we find a $\wh{p}$ such that
  $\|\wh{p}-p\|_{\infty} \leq (2+\eps)\sigma$.

  Finally, we analyze the quality of the initial solution produced by
  $\ell_1$ regression.  By Lemma~\ref{lem:l1reg},
  $\norm{e_0}_1 \leq O(\sigma)$.  Applying the Markov brothers'
  inequality to the degree $d+1$ polynomial
  $Q(x) = \int_{-1}^x e_0(u)du$, we get
  \[
    \norm{e_0}_\infty = \max_{x \in [-1,1]} |Q'(x)| \leq (d+1)^2 \max_{x \in [-1,1]}
    |Q(x)| \leq (d+1)^2\norm{e_0}_1 \leq O(d^2 \sigma).
  \]
  Hence the number of iterations required is $O(\log_{1/\eps} d)$. 
		\end{proof}
Applying this to our random outlier setting, we get Corollary~\ref{cor:randomerrors}.
\begin{proof}
  It is enough to show that for $m=O(\frac{d}{\eps})$, drawing
  $O(\frac{d}{\eps} \log \frac{d}{\delta\eps})$ samples from the
  Chebyshev distribution or
  $O(\frac{d^2}{\eps^2} \log(\frac{1}{\delta}))$ samples from the
  uniform distribution gives us an $\alpha$-good sample for some
  $\alpha< \frac{1}{2}$ with probability $1-\delta$. The corollary
  then follows from Theorem~\ref{thm:main}.

  Let $k$ be the total number of samples taken, and let $p_j$ denote
  the probability that any sampled $x_i$ lies in the $j^{\text{th}}$
  interval $I_j$.  With Chebyshev sampling,
  \[
    p_j = \int_{\cos\left(\frac{\pi j}{m}\right)}^{\cos\left(\frac{\pi
          (j+1)}{m}\right)} \frac{1}{\sqrt{1-x^2}} =
    \arcsin\left(\cos\left(\frac{\pi (j+1)}{m}\right)\right) -
    \arcsin\left(\cos\left(\frac{\pi j}{m}\right)\right) = \frac{1}{m}
  \]
  for all $j$.  With uniform sampling,
  $p_j = \Theta(\frac{\min(j, m+1-j)}{m^2})$.

  Let $X_j$ be the number of samples that appear in $I_j$, and $Y_j$
  be the number of these samples that are outliers.  We have
  $\E[X_j] = k p_j$, so by a Chernoff bound,
  \begin{align}
    \Pr[X_j \leq \frac{1}{2}k p_j] \leq e^{-\Omega(kp_j)}.\label{eq:6}
  \end{align}
  The outliers are then chosen independently, with expectation
  $\rho X_j$, so
  \begin{align}
    \Pr[Y_j \geq \alpha X_j \mid X_j] \leq e^{-\Omega((\alpha - \rho)^2X_j)}.\label{eq:7}
  \end{align}
  Setting $\alpha = \frac{\rho + \frac{1}{2}}{2}$, we have that
  conditioned on~\eqref{eq:6} not occurring,~\eqref{eq:7} occurs with
  at most $e^{-\Omega(k p_j)}$ probability.  If neither occur, then
  less than an $\alpha$ fraction of the samples in $I_j$ are outliers.
  Hence, by a union bound over the intervals, the samples are
  $\alpha$-good with probability at least
  \begin{align}
    1 - 2 \sum_{i=1}^m e^{-\Omega(k p_j)}.\label{eq:8}
  \end{align}
  In the Chebyshev setting, we have $p_j = 1/m$ and
  $k = O(m \log(m/\delta))$, making~\eqref{eq:8} at least $1-2\delta$
  for appropriately chosen constants.  In the uniform setting, we
  similarly have
  \[
    \sum_{j=1}^m e^{-\Omega(kp_j)} = \sum_{j=1}^m
    e^{-\Omega(m^2\log(1/\delta)\cdot \frac{\min(j, 1+m-j)}{m^2})} \leq 2\sum_{j=1}^{m/2} \delta^j \leq 3\delta,
  \]
  making~\eqref{eq:8} at least $1-6\delta$.  Rescaling $\delta$ gives
  the result, that the samples are $\alpha$-good for some
  $\alpha < 1/2$ with probability at least $1-\delta$.
\end{proof}

\section{Impossibility results}

\subsection{Sample Complexity}\label{sec:sampcomp}
\begin{lemma}
  If the $x_i$ are sampled uniformly then it is not possible to get an
  $O(1)$ approximation in $\ell_{\infty}$ norm to the original
  function in $o(d^2)$ samples and $1/4$ failure probability.
\end{lemma}
\begin{proof}
  Consider an algorithm that gives a $C$-approximation given $s$
  uniform samples with noise level $\sigma = 1$ and zero outliers, for
  $C = O(1)$ and $s = o(d^2)$.  We will construct two polynomials with
  $\ell_\infty$ distance more than $2C$, but for which the samples
  have a constant chance of being indistinguishable.

  Define the polynomials $g(x) = 0$ and
  $f(x) := T_d\left(x+\frac{\alpha}{d^2}\right)$ where $T_d$ is the
  degree $d$ Chebyshev polynomial of the first kind, for some
  $d \geq 4$ and constant $\alpha = 4 \sqrt{2(C-1)}$. By construction,
  $\abs{f(x)} \leq 1$ for
  $x \in \left[-1,1-\frac{\alpha}{d^2}\right]$, but $|f(1)| > 2C$
  because for $d \geq 4$
  \[|f(1)| = \left|T_d\left(1+\frac{\alpha}{d^2}\right)\right| =
    \left|\cosh\left(d \text{ arcosh}\left(1 +
          \frac{\alpha}{d^2}\right)\right)\right| > 2C.\] The final
  inequality above follows from the fact that
  $\cosh(d \text{ arcosh}(1 + \frac{x}{d^2})) \geq 1.9 (1 + x^2/8)$ at
  $d = 4$ and this function is increasing in both $d$ and $x$.  Since
  $\norm{f - g}_\infty > 2C$, no single answer can be a valid
  $C$-approximate recovery of both $f$ and $g$.

  Suppose one always observes samples of the form
  $(x_i, 0)$.  These samples are within $\sigma = 1$ of
  both $f(x)$ and $g(x)$ if they lie in the region
  $[-1, 1-\frac{\alpha}{d^2}]$.  The chance that all samples $x_i$ lie
  in this region is at least
  $(1 - \frac{\alpha}{2d^2})^s \geq e^{-s\alpha/d^2}$, which is
  $e^{-o(1)} > 1/2$.  Hence there is at least a $1/2$ chance that the
  samples from $f$ and $g$ are indistinguishable, so the algorithm has
  a failure probability more than $1/4$.
\end{proof}

	\begin{figure}
	\centering    
	\begin{tikzpicture}[xscale=3,yscale=3,domain=-pi:pi,samples=1000]
	\draw[-] (-1,0) -- (1,0) ;
	\draw[-] (0,-0.25) -- (0,1) ;
	\draw[black] plot ({cos(\x r)},{min(1, (-1/51)*((cos(51*\x r))/cos(\x r)))});
	\draw[dotted] plot({cos(\x r)}, {1});
	\draw (1,1) node[right] {$1$};
	\end{tikzpicture}
	\captionsetup{aboveskip=5pt}
	\caption{$p_b(x)$ for $d=50$ and $b = 0$.}
	\label{fig:fd50}
\end{figure}
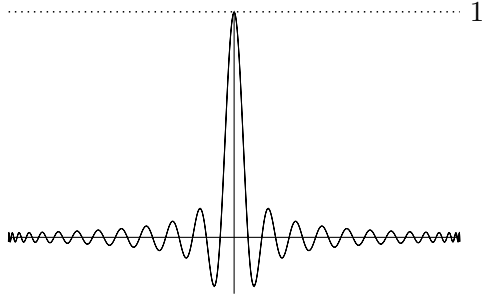

Our next lower bound will use the following lemmas, proven in
Appendix~\ref{app:indicator}.

\define{lem:indicatorpolynomial}{Lemma}{%
  Let $d \geq 1$.  For any point $b \in [-1, 1]$, there exists a
  degree $d$ polynomial $p_b$ such that $p_b(b) = 1$,
  $\norm{p_b}_\infty = 1$, and
  \[
    \abs{p_b(x)} \leq \frac{2}{d\abs{x - b}}
  \]
  for all $x \in [-1, 1]$.
}
\state{lem:indicatorpolynomial}

\define{lem:indicatorpolynomial2}{Lemma}{%
  For any $d$ and $\alpha > 0$, let $m = d\sqrt{\alpha}/2$, and
  define $b_j = -1 + \frac{2}{m} j$ for $j \in [m]$.  Consider the set of degree-$d$
  polynomials
  \[
    f_S(x) = \sum_{j \in S} p_{b_j}^2(x)
  \]
  for $S \subseteq [m]$.  For any $x \in [-1, 1]$, let $k_x \in [m]$ minimize
  $\abs{b_{k_x} - x}$.  Then for any $S \subseteq [m]$,
  \[
   f_{\{k_x\}\cap S}(x) \leq  f_S(x) \leq f_{\{k_x\}\cap S}(x) + \alpha.
  \]
}
\state{lem:indicatorpolynomial2}

\begin{lemma}
  For any distribution on sets $x=(x_1, \dotsc, x_s)$ of
  $s=o(d\log d)$ sample points with independent outlier chance
  $\rho = \Omega(1)$, it is not possible to get a robust $O(1)$
  approximation in $\ell_{\infty}$ norm to the original function with
  $1/4$ failure probability.
\end{lemma}
\begin{proof}
  Let $m = d/\sqrt{12C}$, and $f_S$ and $k_x$ be as in
  Lemma~\ref{lem:indicatorpolynomial2} for $\alpha = \frac{1}{3C}$.  For any
  $x$, let $L_j := \{i\in[s] \mid k_{x_i} = j\}$.  Since these sets
  are disjoint, there must exist at least $m/2$ different $j$ for
  which $\abs{L_j} \leq 2s/m = o(\log d)$.  Let $B \subseteq [m]$
  contain these $j$.
  We say that a given $L_j$ is ``outlier-full'' if all of the
  $x_i \in L_j$ are outliers, which for $j \in B$ happens with
  probability at least
  \[
    \rho^{\abs{L_j}} \geq 2^{-o(\log d)} \geq 1/\sqrt{d}.
  \]
  Hence the probability that at least one $L_j$ is outlier-full is at least
  \[
    1 - (1 - 1/\sqrt{d})^{\abs{B}} \geq 1 - e^{-m/(2 \sqrt{d})} > 0.99.
  \]

  Suppose the true polynomial $p$ to be learned is $f_S$ for a
  uniformly random $S$, and consider the following adversary.  If at
  least one $L_j$ is outlier-full, she arbitrarily picks one such
  $j^*$ and sets $S'$ to the symmetric difference of $S$ and
  $\{j^*\}$.  She then flips a coin, and with 50\% probability outputs
  $(x_i, f_S(x_i))$ for each $i$, and otherwise outputs
  $(x_i, f_{S'}(x_i))$ for each $i$.  This is valid for
  $\sigma = \frac{1}{3C}$, because for each $i \in [s]$, either
  $k_{x_i} = j^*$ (in which case $x_i \in L_{j^*}$ is an outlier) or
  $\{k_{x_i}\} \cap S = \{k_{x_i}\} \cap S'$ (in which case
  Lemma~\ref{lem:indicatorpolynomial2} implies
  $\abs{f_S(x_i) - f_{S'}(x_i)} \leq \alpha = \frac{1}{3C}$).

  Because the distribution on $j^*$ is independent of $S$, the
  distribution of $S'$ is also uniform, so the algorithm cannot
  distinguish whether it received $f_S$ or $f_{S'}$.  But
  $\norm{f_S - f_{S'}}_\infty = \norm{f_{\{i\}}}_\infty = 1 >
  2C\sigma$, so the algorithm's output cannot satisfy both cases
  simultaneously.  Hence, in the $99\%$ of cases where one $L_j$ is
  outlier-full, the algorithm will have $50\%$ failure probability,
  for $49.5\% > 1/4$ overall.
\end{proof}

\subsection{Approximation Factor}

Any algorithm that relies on the result of $\ell_{\infty}$ projection cannot do significantly better.  There are two functions $p(x)$ and $f(x)$ such that $|p(x)-f(x)| \leq \sigma$ for all $ x\in [-1,1]$, however the $\ell_{\infty}$ projection to the space of all degree $1$ polynomials is almost $2\sigma$ away in $\ell_{\infty}$ norm from $p$. 


\begin{lemma}\label{algotight} 
Let $d=1$, $p(x) = \sigma$ be a constant function and $f(x) = \frac{\sigma}{\alpha}\max\left(0, x - (1-2\alpha)\right)$ where $\alpha < \frac{1}{2}$. Note that $|p(x)-f(x)| \leq \sigma$ for all $ x\in [-1,1]$. The result $q = \argmin_{r} \| f - r \|_{\infty}$ of $\ell_{\infty}$ projection of $f$ to the space of degree-$1$ polynomials satisfies $\|p(x)-q(x)\|_{[-1,1],\infty} \geq (2-\alpha)\sigma$.
\end{lemma}
\begin{proof}
 Observe that $f(x) = 0$ for $x \in [-1,(1-2\alpha)]$ and $f(x) = 2\left( \frac{\sigma}{\alpha}(x-1)+\sigma\right)$ in $[1-2\alpha, 1]$. The maximum difference between two linear equations in a closed interval is attained at the endpoints of that interval. This tells us 
 \[q(x) = \argmin_{r} \max \{|f(-1)-r(-1)|,|f(1-2\alpha)-r(1-2\alpha)|,|f(1)-r(1)|\} \] \[= \argmin_{r} \max \{|r(-1)|,|r(1-2\alpha)|,|2\sigma-r(1)|\}.\]
  Let $q(x) = ax + b$. $q(x)$ will be such that $f(-1) > q(-1)$,
 $f(1-2\alpha) < q(1-2\alpha)$ and $f(1) > q(1)$ (see Figure 4), and so we
 want \[\argmin_{(a,b)} \max \{a-b, a+b-2\alpha a,2\sigma-(a+b)\}.\] 
 This function is
 minimized when all three terms are equal, which happens when
 $a = \sigma$ and $b = -\alpha\sigma$.  This gives
 $q(x) = 2\sigma x - \alpha\sigma$, and $\|q(x)-p(x)\|_{[-1,1],\infty} =   (2-\alpha)\sigma$. 

 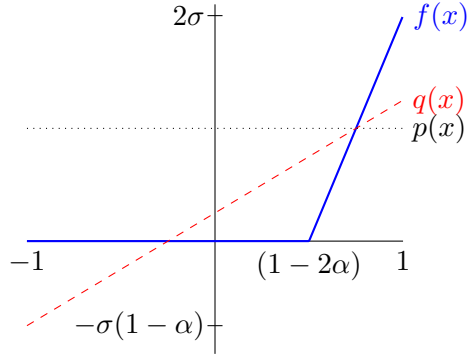
\begin{figure}
   \centering
   \begin{tikzpicture}[xscale=2.5,yscale=1.5,domain=-1:1,samples=400]
     \draw[-] (-1,0) -- (1,0) ;
     \draw[-] (0,-1) -- (0,2.1) ;
     \draw[blue, thick] plot (\x,{4*max(0, \x - (1-1/2)) )}) node[right] {$f(x)$};
     \draw[red, dashed] plot (\x,{\x + 1/4}) node[right] {$q(x)$};
     \draw[dotted] plot(\x, {1}) node[right] {$p(x)$};
     \draw (1/2,0) node[below] {$(1-2\alpha)$};
     \draw (1,0) node[below] {$1$};
     \draw (-1,0) node[below] {$-1$};
     \draw (0,2) node[left] {$2\sigma$};
     \draw(0,-3/4) node[left] {$-\sigma(1-\alpha)$};
     \draw (-.04, -3/4) -- (.04, -3/4);
     \draw (-.04, 2) -- (.04, 2);
   \end{tikzpicture}
   \captionsetup{aboveskip=5pt}
   \caption{$q(x)$ makes an error of $(2-\alpha)\sigma$ with $p(x)$.}
   \label{fig:linetight}
 \end{figure}
\end{proof}

We now show that that one cannot hope for a proper
$(1+\eps)$-approximate algorithm, even with no outliers.  We will
present a set of polynomials such that no two are more than $2$-apart,
but for which no single polynomial lies within $\alpha > 1$ of all
polynomials in the set.  Then an adversary with $\sigma = 1$ can
output a function $y(x)$ independent of the choice of polynomial in
the set, forcing the algorithm to by $\alpha$-far when recovering some
polynomial in the set.

\begin{lemma}
  There exist three degree $\leq 2$ polynomials, all within $2$ of each other
  over $[-1, 1]$, such that any single quadratic function has
  $\ell_\infty$ distance more than $1.09$ from one of the three.
\end{lemma}
\begin{proof}
  Consider the polynomials $p_1(x) = (x+1)$, $p_2(x) = (1-x)$,
  $p_3(x) = \frac{3+2\sqrt{2}}{2}(1-x^2)$, and let $v$ denote
  $\frac{1}{3+2\sqrt{2}}$ (see
  Figure~\ref{fig:counterexample}). Observe that $p_1$ and $p_2$ are
  at distance 2 from each other at $x = 1, -1$. Also $p_1$ is at
  distance $2$ from $p_3$ at $x = -v$, and similarly $p_2$ is at
  distance $2$ from $p_3$ at $x = v$. Hence, any polynomial that wants
  to $1$-approximate all the $p_i$'s necessarily has to go through the
  points $(1,1), (-1,1), (v,2-v), (-v,2-v)$. By symmetry, we know that
  any quadratic that goes through these will be of the form
  $y = ax^2 + c$. Substituting these values in the equation and
  solving for $a$ and $c$ we see that $a = \frac{1}{v+1}$ and
  $c = 1-\frac{1}{v+1}$.

  Since the quadratic $ax^2 +c$ has to $1$-approximate the $p_i$ it
  must be the case that $ax^2 + c \leq p_1(x)+1 = x+2$ at all
  points. However this inequality is not satisfied in the interval
  $(-1,-v)$ as shown in Figure~\ref{fig:counterexample}. This is
  because $p_1(x)+1$ is the line between two points on the curve
  $ax^2 + c$, which is a concave function for $a < 0$.  Running a
  program to optimize parameters, we see that the best approximation
  to these polynomials will make an error of at least $1.09$ with one
  of the polynomials (see Figure~\ref{fig:lowerbound}).

  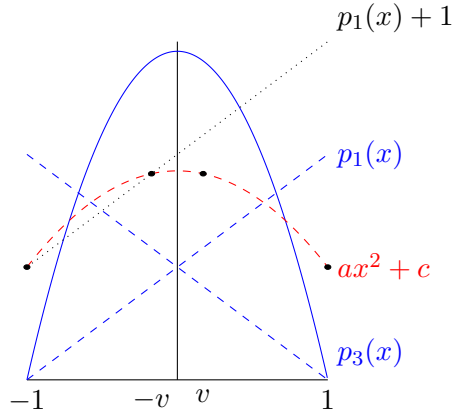
\begin{figure}
    \centering
    \begin{tikzpicture}[xscale=2,yscale=1.5,domain=-1:1,samples=400]
      \draw[-] (-1,0) -- (1,0) ;
      \draw[-] (0,0) -- (0,3) ;
      \draw[blue, dashed] plot (\x,{\x+1}) node[right] {$p_1(x)$};
      \draw[blue, dashed] plot (\x,{1-\x});
      \draw[blue] plot(\x, {2.914213*(1-\x*\x)}) node[above right] {$p_3(x)$};
      \draw[red,dashed] plot(\x, {-0.8535*\x*\x + 1.8535}) node[right] {$ax^2+c$};
      \draw[black,dotted] plot(\x, {2+\x}) node[above right] {$p_1(x)+1$};
      \draw (-0.17157,0) node[below] {$-v$};
      \draw (0.17157,0) node[below] {$v$};
      \draw (1,0) node[below] {$1$};
      \draw (-1, 0) node[below] {$-1$};
      \foreach \Point in {(-1,1), (1,1), (-0.17157,1.82842), (0.17157, 1.82842)} {\draw[fill=black] \Point circle (0.02);}
    \end{tikzpicture}
    \captionsetup{aboveskip=5pt}
    \caption{The quadratic that passes through the four
      points above does not satisfy the inequality
      $ax^2 + c \leq p_1(x)+1$ in the range $[-1,-v]$, for
      $v = \frac{1}{3+2\sqrt{2}}$.}
    \label{fig:counterexample}
  \end{figure}
\end{proof}

	\begin{figure}
		\centering
		\begin{tikzpicture}[xscale=3,yscale=3,domain=-1:1,samples=400]
		\draw[-] (-1,0) -- (1,0) ;
		\draw[-] (0,-.2) -- (0,1.5) ;
		\draw[red, dashed] plot (\x,{ (1/16)*(4*\x+1)^2});
		\draw[red, dashed] plot(\x, {(1/16)*(4*\x-1)^2});
		\draw[black] plot(\x, {(1/16)*(16*(\x)^2 + 1/2)});
		\draw[black] plot(\x, {(1/16)*(-(4*\x)^2 + 4^2)});
		\draw[dotted] (1/8,0) -- (1/8,1.5);
		\draw[dotted] (-1/8,0)--(-1/8,1.5);
	    \draw[dotted] plot(\x, {(1/16)*(-(4*(1/8))^2 + 4^2)});
		\draw[dotted] plot(\x, {(1/16)*(4*(1/8)-1)^2});
		\foreach \Point in {(1/8,(1/2-1/96)} {\draw[fill=black] \Point circle (0.02);}
	 	\foreach \Point in {(-1/8,1/2-1/96)}{\draw[fill=black] \Point circle (0.02);}
 		\foreach \Point in {(0,1/2+1/48)}{\draw[fill=black] \Point circle (0.02);}
		\end{tikzpicture}
		\hspace{1cm}
		\begin{tikzpicture}[xscale=3,yscale=3,domain=-1:1,samples=400]
                  \begin{scope}
                \clip  (-1,-0.25) rectangle (1,1.25);
		\draw[-] (-1,0) -- (1,0) ;
		\draw[-] (0,-0.25) -- (0,1.5) ;
		\draw[red, dashed] plot (\x,{ (1/16)*(4*\x+1)^2});
		\draw[red, dashed] plot(\x, {(1/16)*(4*\x-1)^2});
		\draw[black] plot(\x, {(1/16)*(16*(\x)^2 + 1/2)});
		\draw[black] plot(\x, {(1/16)*(-(4*\x)^2 + 4^2)});
		\draw[red, dashed] plot (\x,{ (1/16)*(4*(\x-0.5)+1)^2});
		\draw[red, dashed] plot(\x, {(1/16)*(4*(\x-0.5)-1)^2});
		\draw[black] plot(\x, {(1/16)*(16*((\x-0.5))^2 + 1/2)});
		\draw[black] plot(\x, {(1/16)*(-(4*(\x-0.5))^2 + 4^2)});
		\draw[red, dashed] plot (\x,{ (1/16)*(4*(\x+0.5)+1)^2});
		\draw[red, dashed] plot(\x, {(1/16)*(4*(\x+0.5)-1)^2});
		\draw[black] plot(\x, {(1/16)*(16*((\x+0.5))^2 + 1/2)});
		\draw[black] plot(\x, {(1/16)*(-(4*(\x+0.5))^2 + 4^2)});
		\foreach \Point in {(1/8,1/2-1/96)} {\draw[fill=black] \Point circle (0.02);}
		\foreach \Point in {(-1/8,1/2-1/96)}{\draw[fill=black] \Point circle (0.02);}
		\foreach \Point in {(0,1/2+1/48)}{\draw[fill=black] \Point circle (0.02);}
				\foreach \Point in {(1/8+0.5,1*1/2-1/96)} {\draw[fill=black] \Point circle (0.02);}
		\foreach \Point in {(-1/8+0.5,1/2-1/96)}{\draw[fill=black] \Point circle (0.02);}
		\foreach \Point in {(0.5,1/2+1/48)}{\draw[fill=black] \Point circle (0.02);}
				\foreach \Point in {(1/8-0.5,1*1/2-1/96)} {\draw[fill=black] \Point circle (0.02);}
		\foreach \Point in {(-1/8-0.5,1/2-1/96)}{\draw[fill=black] \Point circle (0.02);}
		\foreach \Point in {(-0.5,1/2+1/48)}{\draw[fill=black] \Point circle (0.02);}
                  \end{scope}
		
		\end{tikzpicture}
		\captionsetup{aboveskip=5pt}
		\caption{Construction for Lemma~\ref{lem:dapprox}. On
                  the left, any $1$-approximator necessarily has to go
                  through the three points. On the right, placing
                  translated copies of these functions force any
                  $1$-approximation to perform more than $d$
                  oscillations. These figures are not to scale. }
		\label{fig:degdlower}
	\end{figure}
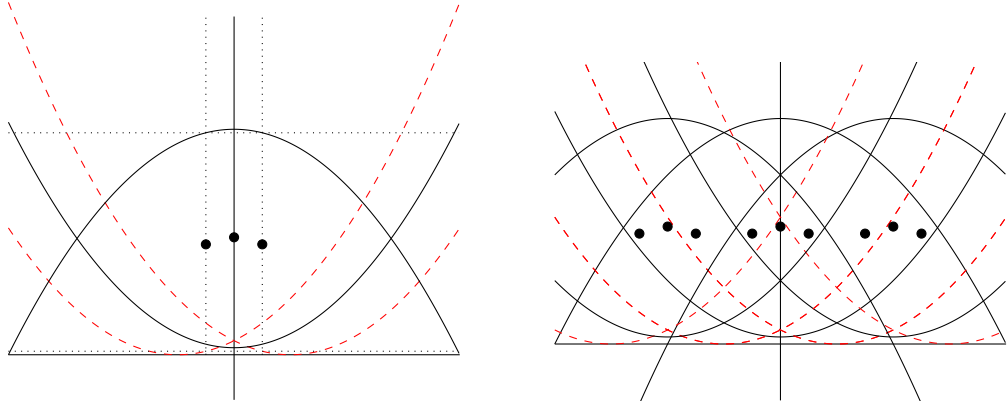

        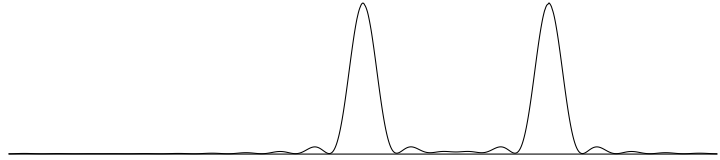
\begin{figure}
  \centering
  \begin{tikzpicture}[xscale=3,yscale=2,domain=0:pi, samples=1010]
    \draw[black] plot ({\x},{min(1, (1/21^2)*(((cos(21*\x r))/cos(\x r))^2)) + min(1, (1/21^2)*((cos(21*(\x + 31*pi/42) r))/cos((\x + 31*pi/42) r))^2)});
    \draw[-] (0,0) -- (pi,0) ;
  \end{tikzpicture}
  \caption{One element of $f_{S}$ from Lemma~\ref{lem:indicatorpolynomial2}.}
  \label{fig:Indicatorsum}
\end{figure}

\begin{lemma}\label{lem:dapprox}
  There exist $O(d)$ degree $2$ polynomials such that any degree $d$ polynomial that tries to approximate all these polynomials has to make an error of $\Omega(\frac{1}{d^3})$ with at least one of the degree $2$ polynomials.  
\end{lemma}
\begin{proof} 
	We will consider a set of quadratic equations such that any two are at most distance $\sigma$ from each other, and such that any polynomial that has to $\sigma$-approximate them needs to perform $2d$ oscillations. As no degree $d$ polynomial can perform more than $d$ oscillations, there is at least one oscillation that this polynomial does not perform, and hence the approximating polynomial will make an error of at least $1+h$ where $h$ is the height of the oscillation. 
	
	 Observe that $\| (1-Ax^2) - A(x\pm c)^2\|_{\infty} = 1-\frac{Ac^2}{2}$ and this is achieved at $x = \pm \frac{c}{2}$. This means the set $S_0 = \{ 1-Ax^2, A(x-c)^2, A(x+c)^2, Ax^2 + \frac{Ac^2}{2}\}$ has every polynomial at most at distance $2\sigma := 1-\frac{Ac^2}{2}$ from each other. Any $\sigma$ approximator, hence, necessarily has to go through the three points as shown in Figure~\ref{fig:degdlower}. Define $S_t = \{ 1-A(x-2ct)^2, A(x-2ct-c)^2, A(x-2ct+c)^2, A(x-2ct)^2 + \frac{Ac^2}{2}\}$ Now observe that any $\sigma$ approximator to $S = \cup_{t \in [-1.5\cdot d,1.5 \cdot d]} S_t$ will necessarily have to perform $2d$ oscillations. We will now define the parameters such that these $2d$ oscillations take place in the range $[-1,1]$ and every pair of functions is at distance at most $1-\frac{Ac^2}{2}$ from each other.  Set $A = \frac{1}{2d}$ and $c = \frac{1}{4d}$. This ensures that every pair of elements in $S_t$ are at most at distance $1-\frac{1}{64d^3}$ from each other.
	 
	 We now show that if $p \in S_t$ and $q \in S_{t'}$ for $t < t'$, then $\|p - q\|_{\infty} \leq 1-\frac{1}{64d^3}$. Observe that it is enough to check this for $p = 1-A(x-2ct)^2$ and $q = A(x-2ct'+c)^2$. Because of our choice of $A,c$ 
	 \begin{align*}
	\left| 1-A(x-2ct)^2 - A(x-2ct'+c)^2 \right| 
	&= \left|1 - \frac{1}{2d} \left( \left(x-\frac{2t}{4d}\right)^2 + \left(x - \frac{2t'-1}{4d}\right)^2 \right)\right| \\
	&  \leq \left| 1 - \frac{1}{2d}\left(\frac{2}{4} \cdot \frac{(2(t'-t) -1)^2}{16d^2}\right)\right| \\
	& \leq 1 - \frac{1}{64d^3}
	 \end{align*}
	 Finally, observe that we force the approximating polnomial to perform one oscillation for every $S_t$, and if $c = \frac{1}{4d}$ the polynomial has to perform $\frac{8d}{3}$ oscillations to $\sigma$-approximate every polynomial in $S$ in the interval $[-1,1]$ because it has to perform one oscillation in every interval of length $3c$. Since no degree $d$ polynomial can perform $2d$ oscillations, there is at least one oscillation that it cannot perform, and so the approximating polynomial necessarily has to make an error of at least $1 + h$ with one of the polynomials in $S$, where $h$ is the height of the oscillation, which is $\Omega(\frac{1}{d^3})$.

\end{proof}

\subsection{List decoding}

We now show that if the probability of getting a bad sample were
greater than $\frac{1}{2}$, then it is not possible to find $\poly(d)$
polynomials of degree $O(d)$ such that one of the polynomials is close
to the original polynomial.

\begin{theorem}\label{thm:listdecode}
  Consider any algorithm for robust polynomial regression that returns
  a set $L$ of polynomials from samples with outlier chance
  $\rho = 1/2$, such that at least one element of $L$ is an
  $C$-approximation to the true answer with $3/4$ probability.  Then
  $\E[\abs{L}] \geq \frac{3}{4} 2^{\Omega(d/\sqrt{C})}$.
\end{theorem}
\begin{proof}
%
  Note that we may assume $d/C$ is larger than a sufficiently large
  constant, since otherwise the result follows from the fact that
  $\abs{L} \geq 1$ whenever the algorithm succeeds.  Let us then set
  $m = d/\sqrt{12C}$, and take $f_S$ from
  Lemma~\ref{lem:indicatorpolynomial2} for $\alpha = 1/(3C)$.  For any
  $S \subseteq [m]$ and $x \in [-1, 1]$, the lemma implies either
  $f_{\{\}}(x) \leq f_S(x) \leq f_{\{\}}(x) + \alpha$ or
  $f_{[m]}(x) - \alpha \leq f_{S}(x) \leq f_{[m]}(x)$.

  Therefore, when observing any polynomial $f_S$ for
  $S \subseteq [m]$, the adversary for $\sigma = \alpha$ can ensure
  that any sample point $x$ is observed as $f_{\{\}}(x)$ with $50\%$
  probability, and $f_{[m]}(x)$ with $50\%$ probability.  This is
  because $p(x)$ is always within $\alpha$ of one of these, so the
  adversary can output that one if $x$ is an inlier, and the other one
  if $x$ is an outlier.  For this adversary, the algorithm's input is
  independent of the choice of $f_S$.  Hence its distribution on
  output $L$ is also independent of $f_S$.  But each $\wh{p} \in L$
  can only be a $C$-approximation to at most one $f_S$.  Hence, if
  $S \subseteq [m]$ is chosen at random,
  \[
    \Pr[\exists \wh{p} \in L : \norm{\wh{p} - f_S} \leq C \sigma] \leq \frac{\E \abs{L}}{2^m}.
  \]
  This implies the result.
\end{proof}

\section*{Acknowledgements}
We would like to thank user111 on MathOverflow~\cite{u111} for
pointing us to~\cite{N79}.

\bibliographystyle{alpha}
\bibliography{main}

\appendix\label{proof:polyconstant}
\section{Proof of Lemma~\ref{lem:polyconstant}}
To prove Lemma~\ref{lem:polyconstant}, we need a a generalization of
Bernstein's inequality from $\ell_\infty$ to $\ell_q$ for all $q > 0$,
which appears as Theorem 5 of~\cite{N79}.  In the univariate case and
setting its parameters $\gamma, \Gamma = 0$, it states
\begin{lemma}[Special case of Theorem 5 of~\cite{N79}]\label{lem:nevai}
  There exists a universal constant $C$ such that, for any degree $d$
  polynomial $p(x)$ and any $q > 0$,
  \[
  \int_{-1}^1 \abs{\sqrt{1-x^2}p'(x)}^qdx \leq C d^q \int_{-1}^1 \abs{p(x)}^qdx.
  \]
\end{lemma}

\restate{lem:polyconstant}
\begin{proof}
  For any individual $I_k$ and $x_k^* \in I_k$, we have by H\"older's inequality that
  \begin{align*}
    \int_{x \in I_k} \abs{p(x) - p(x^*_k)}^qdx &= \int_{x \in I_k} \tabs{\int_{x^*}^x p'(y)dy}^qdx
    \\&\leq \int_{x \in I_k} \abs{x - x^*}^{q-1}\int_{x^*}^x \abs{p'(y)}^q dydx
    \\&\leq \abs{I_k}^q \int_{x \in I_k} \abs{p'(x)}^qdx.
  \end{align*}
  Hence
  \begin{align}\label{eq:Ione}
    \int_{x \in I_k} \abs{p(x) - r(x)}^qdx  \leq \abs{I_k}^q \int_{x \in I_k} \abs{p'(x)}^qdx.
  \end{align}
  We first consider $I_2, \dotsc I_{m/2}$, then separately consider
  the first interval $I_1$. The statement holds for the intervals $I_{m/2}, \dots, I_m$ by symmetry.  For each interval $I_k$
  with $2 \leq k \leq \frac{m}{2}$, we have for all $ x \in I_k$
  \[
    \abs{I_k} = \cos\left(\frac{k\pi}{m}\right) - \cos\left(\frac{(k+1)\pi}{m}\right) \lesssim \frac{k}{m^2} = \frac{k/m}{m} \lesssim \frac{\sin(k\pi/m)}{m} \lesssim \frac{\sqrt{1-x^2}}{m}
  \]
  where we use the notation $a \lesssim b$ to denote that there exists
  a universal constant $C$ such that $a \leq Cb$.  Hence
  \begin{align*}
    \int_{x \in I_k} \abs{p(x) - r(x)}^qdx &\leq \abs{I_k}^q \int_{x \in I_k} \abs{p'(x)}^qdx
    \lesssim \frac{1}{(cm)^q}\int_{x \in I_k} \abs{\sqrt{1 - x^2}p'(x)}^qdx
  \end{align*}
  and so by Lemma~\ref{lem:nevai}, for some constant $c$,
  \begin{align}
    \int_{x \in I_2 \cup \dotsb \cup I_{m/2}} \abs{p(x) - r(x)}^qdx
    \leq\frac{1}{(cm)^q} \int_{-1}^1 \abs{\sqrt{1 - x^2}p'(x)}^qdx
    \lesssim \left(\frac{d}{m} \norm{p}_q\right)^q.\label{eq:Imiddle}
  \end{align}
  Now we consider the end $I_1$.  By the Markov brothers'
  inequality~\cite{achieser2013theory},
  \[
  \norm{p'}_\infty \leq d^2 \norm{p}_\infty.
  \]
  Let $x^*\in [-1, 1]$ such that $\abs{p(x^*)} = \norm{p}_\infty$, and let
  $I' = \{y \in [-1, 1] \mid \abs{x^*-y} \leq \frac{1}{2d^2}\}$.  We
  have $p(y) \geq \norm{p}_\infty/2$ for all $y \in I'$, so
  \[
  \norm{p}_q^q \geq \abs{I'}(\norm{p}_\infty/2)^q \geq \frac{\norm{p}_\infty^q}{2d^22^q}.
  \]
  Hence by~\eqref{eq:Ione}, and using that $\abs{I_1} = 1 - \cos\frac{\pi}{m} = \Theta(\frac{1}{m^2})$,
  \[
  \int_{x \in I_1} \abs{p(x) - r(x)}^qdx \leq \abs{I_1}^{q+1} \norm{p'}_\infty^q \leq \frac{d^{2q}}{(cm)^{2q+2}}\norm{p}_\infty^q \lesssim \left(\frac{d\sqrt{2}}{cm}\right)^{2q+2} \norm{p}_q^q
  \]
  For some constant $c$. The same holds for intervals $I_{m/2}, \dots, I_m$ by symmetry, so combining with~\eqref{eq:Imiddle} gives
  \[
  \int_{-1}^1 \abs{p(x) - r(x)}^qdx \lesssim \left(\frac{d}{cm}\right)^q\left(1 + 2^{q+1}\left(\frac{d}{m}\right)^2\right) \norm{p}_q^q
  \]
  or
  \[
  \norm{p-r}_q \lesssim \frac{d}{m}\left(1 + \left(\frac{d}{m}\right)^{2/q}\right)\norm{p}_q.
  \]
  When $m \geq d$, the first term dominates giving the result.
\end{proof}

\section{Proof of Lemma~\ref{lem:indicatorpolynomial}}\label{app:indicator}
\restate{lem:indicatorpolynomial}
\begin{proof}
  We will show a stronger form of this lemma for even $d$ and $b = 0$, giving
  \begin{align}
    \abs{p_0(x)} \leq \frac{1}{(d+1)\abs{x}}.\label{eq:9}
  \end{align}
  By subtracting $1$ from odd $d$, this implies the same for general
  $d$ with a $1/d$ rather than $1/(d+1)$ term.  Then for general $b$
  we take $p_b(x) = p_0((x - b)/2)$, which satisfies the lemma.  So it
  suffices to show~\eqref{eq:9} for even $d$ and $b=0$.

  We choose $p_0(x)$ to be
  \[
    p(x) := (-1)^{d/2}\frac{T_{d+1}(x)}{(d+1)x}
  \]
  so that
  \[
    p(\cos \theta) = (-1)^{d/2} \frac{\cos((d+1)\theta)}{(d+1)\cos \theta}
  \]
  or, replacing $\theta$ with $\frac{\pi}{2} - \theta$ and using that $\sin(d\frac{\pi}{2} + \psi) = (-1)^{d/2} \sin \psi$,
  \[
    p(\sin \theta) = \frac{\sin((d+1)\theta)}{(d+1)\sin \theta}.
  \]
  Since $\abs{\sin((d+1) \theta)} \leq 1$, this immediately
  gives~\eqref{eq:9}; we just need to show
  $p(0) = \norm{p}_\infty = 1$.  By L'H\^opital's rule, we have
  \[
    p(0) = \frac{(d+1) \cos((d+1) \cdot 0)}{(d+1) \cos 0} = 1.
  \]
  If $\theta \geq 1.1/(d+1)$, then $\abs{\sin \theta} \geq 1/(d+1)$,
  and so~\eqref{eq:9} implies $\abs{p(\sin \theta)} \leq 1$.  Since
  $p$ is symmetric, all that remains is to show
  $\abs{p(\sin \theta)} \leq 1$ for $0 < \theta < 1.1/(d+1)$.  

  The maximum value of $\abs{p(\sin \theta)}$ will either appear at an
  endpoint of this interval---which we have already shown is at most
  $1$---or at a zero of the derivative.  We have
  \begin{align*}
    \frac{\partial}{\partial \theta} p(\sin \theta)
    &= \frac{(d+1)\cos((d+1) \theta)}{(d+1) \sin \theta} - \frac{\sin((d+1) \theta)}{((d+1) \sin \theta)^2} \cdot (d+1)\cos \theta\\
    &= \frac{(d+1)\cos((d+1)\theta)\sin \theta - \sin((d+1)\theta)\cos \theta}{(d+1)\sin^2 \theta}.
  \end{align*}
  For all $0 < \psi$, we have the inequalities
  $\psi - \psi^3/6 < \sin \psi < \psi$ and
  $1 - \psi^2/2 < \cos \psi < 1 - \psi^2/2 + \psi^4/24$.  Hence for
  $0 < \theta < 1.1/(d+1)$, the denominator is positive and the
  numerator is less than
  \begin{align*}
    &(d+1)(1 - (d+1)^2\theta^2/2 + (d+1)^4 \theta^4/24) \theta - \left((d+1)\theta - (d+1)^3\theta^3/6\right) (1 - \theta^2/2) \\
    &= \theta^3\left(-(d+1)^3/2 + (d+1)/2 + (d+1)^3/6\right) + \theta^5\left((d+1)^5/24 - (d+1)^3/12\right)\\
    &\leq -\frac{5}{18}(\theta(d+1))^3 + (\theta(d+1))^5/24\\
    &\leq -0.22 (\theta(d+1))^3 < 0.
  \end{align*}
  Thus $p(\sin \theta)$ does not have a local maximum on
  $(0, 1.1/(d+1))$, so $\norm{p}_\infty \leq 1$, finishing the proof.
\end{proof}
\restate{lem:indicatorpolynomial2}
\begin{proof}
  Since
  $f_{S \cup \{k_x\}}(x) = f_{\{k_x\}}(x) + f_{S \setminus
    \{k_x\}}(x)$ for any $S$, it is sufficient to prove the result
  when $k_x \notin S$.  Then
  \begin{align*}
    f_S(x) \leq \sum_{j \neq k_x} p_{b_j}^2(x) &\leq \sum_{j \neq k_x} \frac{4}{d^2 \abs{x - b_j}^2} \leq \sum_{j \neq k_x} \frac{4}{d^2 ((2\abs{j-k_x}-1)\frac{2}{m})^2}\notag\\
                                 &< \frac{2m^2}{d^2} \sum_{j=1}^\infty\frac{1}{(2j-1)^2}  < \frac{2m^2}{d^2} \frac{\pi^2}{6}\notag\\
                                 &< \alpha
  \end{align*}
  as desired.
\end{proof}

\end{document}